\documentclass[11pt]{article}
\usepackage{amsthm}
\usepackage{mathtools}
\usepackage{fullpage}
\usepackage{thm-restate}
\usepackage{amsmath}
\usepackage{amssymb}
\usepackage{palatino}
\usepackage{setspace}
\usepackage[colorlinks=true, linkcolor=blue, citecolor=blue, urlcolor=blue]{hyperref}

\newtheorem{theorem}{Theorem}[section]
\newtheorem{conjecture}{Conjecture}[section]
\newtheorem{definition}{Definition}[section]
\newtheorem{lemma}{Lemma}[section]
\newtheorem{corollary}{Corollary}
\newtheorem{question}{Question}
\newtheorem{problem}{Open Problem}

\DeclarePairedDelimiter\abs{\lvert}{\rvert}
\DeclarePairedDelimiter\norm{\lVert}{\rVert}
\DeclarePairedDelimiter\ceil{\lceil}{\rceil}

\DeclarePairedDelimiter\ip{\langle}{\rangle}
\DeclarePairedDelimiter\rb{(}{)}

\newcommand{\cclass}[1]{\mathsf{#1}}
\newcommand{\acc}{\cclass{{ACC}}^0}

\newcommand{\ccz}{\cclass{{CC}}^0}
\newcommand{\AND}{\cclass{{AND}}}
\newcommand{\NOT}{\cclass{{NOT}}}
\newcommand{\OR}{\cclass{{OR}}}
\newcommand{\modm}{\cclass{{MOD}}_m}
\newcommand{\modg}{\cclass{{MOD}}}
\newcommand{\maj}{\cclass{{MAJORITY}}}
\newcommand{\modq}{\cclass{{MOD}}_q}
\newcommand{\modp}{\cclass{{MOD}}_p}

\newcommand{\poly}{\cclass{{poly}}}
\newcommand{\aut}{\mathop{{Aut}}}
\newcommand{\nulls}{\mathop{{nullspace}}}

\title{On CC\texorpdfstring{\textsuperscript{0}}{0} Lower Bounds for AND via Torus Polynomials}
\author{Vaibhav Krishan\thanks{The Institute of Mathematical Sciences, Chennai, India. Email:~{\tt vaibhavk@imsc.res.in}} \and Jayalal Sarma\thanks{Indian Institute of Technology Madras (IIT Madras), Chennai, India. Email:~{\tt jayalal@cse.iitm.ac.in}}}

\begin{document}

\maketitle

\begin{abstract}

We explore the torus polynomial approximation based approach towards a long-standing question: whether \( \AND \) can be computed by \(\ccz\) circuits - the class of constant-depth polynomial size circuits containing \(\modm\) gates for some natural number \(m\).
Bhrushundi, Hosseini, Lovett and Rao (ITCS 2019) introduced torus polynomial approximations as an approach for proving lower bounds against \(\acc\) - a class containing \(\ccz\) where the circuits are also allowed \(\AND\), \(\OR\) and \(\NOT\) gates. 

We show how lower bounds for torus polynomials approximating \(\AND\) can be used to make progress on this question. 
Using lower bounds on the degree of symmetric torus polynomials approximating \( \AND \), proved by Krishan and Vishwanathan (ITCS 2026), we prove size lower bounds for \emph{symmetric} \( \ccz \)-circuits computing \( \AND\).
More precisely, we prove that any depth $h$ symmetric \(\ccz\) circuit requires \( 2^{\widetilde{\Omega}(n^{1/O(h)})} \) size to compute \(\AND\). 

A key ingredient in our proof is an argument that we can construct symmetric torus polynomials to approximate symmetric \( \ccz \) circuits.
Our construction exhibits an explicit correspondence between the symmetry of the circuit and that of the polynomial.
Using this, we also establish lower bounds for weaker notions of circuit symmetry.
Lower bounds for symmetric \( \ccz \) circuits were also independently established by Pago (ICALP 2026) using different techniques.

In the \emph{asymmetric} regime, we establish degree upper bounds for depth three circuits of the form \(\modp \circ \modg_{m} \circ \AND_{O(1)}\) where \(m=pq\) is a semiprime.
This circuit class is a special case of the \textit{constant degree hypothesis}, introduced by Barrington, Straubing and Th{\'e}rien (Information and Computation, 1990), where $m$ could be an arbitrary composite number.
We argue that improved lower bounds for asymmetric torus polynomials approximating \(\AND\) imply size lower bounds for semiprime $m$ and hence progress on the constant-degree hypothesis.

\end{abstract}

\newpage
\tableofcontents

\section{Introduction}

The polynomial method has proven to be a powerful tool for tackling fundamental questions in theoretical computer science.
In particular, studying polynomial approximations for Boolean functions has led to advances in cryptography, quantum computing and circuit complexity (see~\cite{Aar08,BT20} and references therein).
In circuit complexity, considering the degree of polynomials approximating Boolean functions has led to remarkable progress in proving circuit lower bounds for explicit functions, a notoriously difficult quest in the area.

Landmark results in this direction were proved by Razborov~\cite{razborov1987lower}, and independently by Smolensky~\cite{smolensky1987algebraic}.
They proved that for any prime $p$, $\modp$\footnote{A \(\modp\) gate outputs \(1\) if and only if the sum of its inputs is not divisible by \(p\).} cannot be computed by constant-depth polynomial size circuits that use \(\AND, \OR, \NOT\) and $\modq$ gates for a prime \(q \neq p\).
In addition, they used the same technique to prove that such circuits cannot compute $\maj$.
One main technical step in these arguments is that any function in the circuit class can be approximated by a low-degree polynomial over the finite field $\mathbb{F}_q$.
This is complemented by an argument that the candidate function requires a high degree for any polynomial that approximates it.

Since then, polynomial approximations have been used extensively for various lower bounds against constant-depth circuits \cite{alon2020unbalancing,hrubes2019lower,limaye2021fixed,oliveira2019parity}, and even for constructing \emph{pseudo-random generators} fooling such circuits~\cite{harsha2019polynomial}.
However, the method has hit a roadblock in showing superpolynomial lower bounds when $\modg_6$ gates are allowed.
More generally, it is unclear whether low-degree polynomials over fields/rings can approximate functions in \(\acc\), the class of constant-depth polynomial size circuits that use \(\modm\) gates for some natural number \(m\).

Recently, Bhrushundi, Hosseini, Lovett and Rao~\cite{bhrushundi2019torus} introduced a new notion of polynomial approximations, called \emph{torus polynomial approximations} (see also~\cite{krishan2021upper}) \footnote{A polynomial $P$ is a torus polynomial that $\varepsilon$-approximates a Boolean function $f$, if for every $x \in \{0,1\}^n$, the fractional part of $P(x)$ is $\varepsilon$-close to $\frac{f(x)}{2}$.}, specifically towards approximating functions in \(\acc\).
They showed that all functions in \(\acc\) have low-degree torus polynomial approximations, establishing it as a plausible method for proving lower bounds against \(\acc\).
However, strong lower bounds are not known on the degree required for approximating any explicit function under this notion of approximation.
In particular, the authors in~\cite{bhrushundi2019torus} conjectured that low-degree torus polynomials cannot approximate $\maj$. 

The authors in~\cite{bhrushundi2019torus} further showed that if we restrict to symmetric polynomials\footnote{A polynomial is symmetric if monomials of the same degree share the same coefficient.}, then $\maj$ requires $\Omega\rb*{\sqrt{\frac{n}{\log n}}}$ degree to approximate it within an error of $\frac{1}{20n}$.
In a subsequent work, Krishan and Vishwanathan~\cite{krishan2026lower} showed that this restriction on symmetry is too strong.
The authors proved, in~\cite[Theorem 7]{krishan2026lower}, that the same degree lower bound holds even for the $\AND$ function.
They use this lower bound to argue that studying symmetric torus polynomials is not suitable for resolving the $\maj$ vs $\acc$ question.
In this work, we use symmetric torus polynomials towards another dual frontier - whether constant-depth polynomial size circuits can compute $\AND$ with only $\modg$ gates.
More precisely:

\begin{question}
    \label{que:cczand}
    Does \(\ccz\) contain the \(\AND\) function?
\end{question}
where \(\ccz = \bigcup_{m \in \mathbb{N}} \ccz[m]\), and \(\ccz[m]\) denotes the class of constant-depth polynomial size circuits comprising \(\modm\) gates.

It is widely believed, especially since the work in~\cite{barrington1990non, barrington1994representing}, that $\AND \notin \ccz$.
This question is well-understood when \(m\) is a power of a prime \(p\).
Using classical results from~\cite{razborov1987lower, smolensky1987algebraic}, one can obtain small-degree polynomials over \(\mathbb{F}_p\) representing\footnote{A polynomial \(P\) represents a function over \(\mathbb{F}_p\) if for each \(x \in \{0, 1\}^n\), \(P(x) = f(x) \mod p\).} such circuits.
On the other hand, \(\AND\) requires a large degree to represent over any \(\mathbb{F}_p\), hence proving the lower bound.

However, progress on the conjecture when $m$ is composite is limited to slightly superlinear lower bounds~\cite{chattopadhyay2006lower} in the general setting, or strong lower bounds in highly restricted settings~\cite{kawalek2025violating, pago2026optimal}.
In fact, a seemingly simple case of Question~\ref{que:cczand} remains open, posed by Barrington, Straubing and Th{\'e}rien ~\cite{barrington1990non}, known as the \emph{constant degree hypothesis}.
They conjectured that any circuit of the form \(\modp \circ \modm \circ \AND_{O(1)}\) requires exponential size to compute \(\AND\).
Here, \(p\) is a prime, and \(\AND_{O(1)}\) denotes an \(\AND\) gate of constant fan-in.

\subsubsection*{Our Results:}
We contribute to the program of proving circuit lower bounds using torus polynomial approximations against symmetric $\ccz$ circuits.
A circuit is said to be \emph{symmetric} if each permutation of its variables can be extended to an automorphism of the underlying DAG.
Kawa\l{}ek and Wie{\ss}~\cite{kawalek2025violating} initiated the study of symmetric \(\ccz\) circuits, and proved a special case of the constant degree hypothesis.
Their lower bound applied to symmetric \(\modp \circ \modq \circ \AND_{O(1)}\) circuits for two primes \(p\) and \(q\).
Soon thereafter, Pago~\cite{pago2026optimal} extended the lower bound to symmetric \(\ccz\) circuits.

Building on the lower bound in~\cite[Theorem 7]{krishan2026lower}, on the degree required to approximate $\AND$ by symmetric torus polynomials, we prove a size lower bound for \emph{symmetric} \(\ccz\) circuits computing \(\AND\). More formally,

\begin{theorem}
\label{thm:symcczlb}
Any symmetric \(\ccz\) circuit of depth \(h\) requires size \(s(n) = {2^{\widetilde{\Omega}(n^{1/O(h)})}}\) to compute \(\AND_n\).
\end{theorem}

We use the degree of symmetric torus polynomials approximating a function within $\frac{1}{20n}$ error as the measure for the lower bound. More precisely, we construct low-degree symmetric torus polynomials that approximate symmetric $\ccz$ circuits within $\frac{1}{20n}$ error. We follow the construction from~\cite[Corollary 20]{bhrushundi2019torus} to obtain the torus polynomials approximating \(\ccz\) circuits. Our main contribution is to show that the construction produces symmetric torus polynomials if we start with symmetric circuits. Then, we use the degree lower bound from~\cite[Theorem 7]{krishan2026lower} to prove a size lower bound for symmetric \(\ccz\) circuits computing \(\AND\).

In an independent line of work towards establishing size lower bounds against symmetric circuits, Kawa\l{}ek and Wie{\ss}~\cite{kawalek2025violating} and Pago~\cite{pago2026optimal} use the period of symmetric functions as the measure.
A symmetric function \(f\) over \(n\) variables has period \(b\) if \(f(i) = f(i+b)\) for each \(0 \leq i \leq n-b\), where \(f(i)\) denotes the function's value at points with Hamming weight \(i\). In both~\cite{kawalek2025violating, pago2026optimal}, the main technical argument is to establish that respective symmetric functions computed by the circuits have small period.
In contrast, \(\AND\) has period exactly \(n\), which leads to lower bounds on the size of the circuit.

We show that the degree of approximation by symmetric torus polynomials subsumes the period of a function as a measure.
Formally, consider the measure \(\mu(f)\) for a Boolean function \(f\), defined as the minimum degree of a symmetric torus polynomial that \(\frac{1}{20n}\)-approximates \(f\).
We show that for any function with a small period, $\mu(f)$ is small.
Note that our statement assumes that the period has only a few distinct prime divisors.
This is indeed true for the period of symmetric \(\ccz\) circuits as described in~\cite{pago2026optimal} (see~\cite[Lemma 3.3]{pago2026optimalarxiv} for the exact expression of the period).

\begin{restatable}{theorem}{thmperiodapprox}
\label{introthm:periodapprox}
    Consider any periodic symmetric function \(f: \{0, 1\}^n \to \{0, 1\}\) with period \(m\), where \(m\) has \(O(1)\) distinct prime divisors.
    Then, there exists a symmetric torus polynomial of degree \(d=O(m \log^{O(1)}(n))\) approximating \(f\) within \(\frac{1}{20n}\) error.
    In other words, \(\mu(f) = \widetilde{O}(m)\), where \(\widetilde{O}\) hides some polylog factors.
\end{restatable}
    
We note that \(\AND\) has period \(n\) but has \(\widetilde{O}(\sqrt{n})\)-degree symmetric torus polynomial approximations (see~\cite[Remark 22]{krishan2026lower}).
Hence, there is a quadratic separation between the period and the degree. 
It is conceivable that the degree measure captures a larger class of functions than periodic functions.
Therefore, studying torus polynomial approximations, and in particular symmetric torus polynomials, may enable progress beyond studying the period of functions.

However, note that the size lower bound we prove in Theorem~\ref{thm:symcczlb} is quantitatively weaker than the corresponding lower bound in~\cite{pago2026optimal}.
The exponent in our result diminishes as the depth of the circuit grows.
On the other hand, in~\cite{pago2026optimal}, they prove a size lower bound of the form \(2^{\Omega\rb{n^{1/r}}}\) for symmetric \(\ccz[m]\) circuits, where \(r\) denotes the number of distinct prime divisors of \(m\).
Crucially, their lower bound does not depend on the circuit's depth.
Theorem~\ref{introthm:periodapprox} shows that, in principle, a similar dependence on \(r\) can be obtained through degree as a measure as well. To see this, we note that in~\cite{pago2026optimal}, they proved that any size-\(s\) symmetric \(\ccz[m]\) circuits computes a symmetric function with period \(O(s^r)\). Moreover, the period has $O(1)$ many distinct prime divisors.
Hence, we can combine Theorem~\ref{introthm:periodapprox}, and the degree lower bound from~\cite{krishan2026lower}, to obtain a size lower bound of the form \(2^{\widetilde{\Omega}\rb{n^{1/2r}}}\).

Finally, our approach provides an advantage when considering relaxed notions of symmetries for \(\ccz\) circuits.
Pago~\cite{pago2026optimal} also studied \(\ccz\) circuits in a more general setting, which they called \emph{nested block symmetry}, where the automorphism group of the circuit extends a particular subgroup of \(S_n\).
Using our techniques, we can derive lower bounds for such circuits as well; see Theorem~\ref{thm:nestedsymlb} for the statement. We note that our approach establishes a correspondence between the circuit's symmetry and that of the approximating polynomial. This makes it easier to establish the lower bound for restricted notions of symmetry of the circuit without introducing an additional measure (compared to \cite{pago2026optimal}).

\noindent {\bf Asymetric Toroidal Approximations:}
Motivated by the above discussion, we propose the degree of torus polynomial approximations as a measure to study the \(\ccz\) vs \(\AND\) question.
In the special case when the polynomial is symmetric, we get a size lower bound for symmetric \(\ccz\) circuits.
Next, we show that progress on lower bounds for general torus polynomials approximating \(\AND\) can lead to progress on an important case of the constant degree hypothesis.
For simplicity, we state the result for \(\modg_2 \circ \modm \circ \AND_{O(1)}\) circuits when \(m\) is an even semiprime, see Theorem~\ref{thm:cdhapprox} for a more general statement.

\begin{theorem}
\label{thm:2cdhapprox}
    Consider any \(\modg_2 \circ \modm \circ \AND_{O(1)}\) circuit \(C\) of polynomial size, where \(m\) is an even semiprime.
    Then, there exists a degree-\(O(\log(n))\) torus polynomial approximating \(C\) within \(\frac{1}{\Omega(n)}\) error.
\end{theorem}

This result suggests a potential approach for proving that \(\modg_2 \circ \modm \circ \AND_{O(1)}\) circuits, for an even semiprime \(n\), require superpolynomial size to compute \(\AND_n\).
Using~\cite[Theorem 6]{krishan2026lower}, we get that any torus polynomial approximating \(\AND\) within \(\frac{1}{\Omega(n)}\) error must have degree \(\Omega(\log(n))\).
The authors note a gap between this lower bound and the known upper bound for \(\AND\), and suggest bridging this gap in~\cite[Open Problem 4]{krishan2026lower}, hinting at a possibility of improving the lower bound.
Using Theorem~\ref{thm:2cdhapprox}, improving the degree lower bound to \(\omega(\log(n))\) implies the superpolynomial size lower bound we are aiming for.
In the appendix, we present evidence that there is room for improvement in the lower bound argument from~\cite[Theorem 6]{krishan2026lower}.
Hence, we conjecture that the lower bound is indeed stronger, see Conjecture~\ref{conj:genandlb} for a more general statement.

\begin{conjecture}
    \label{conj:2andlb}
    Any torus polynomial that approximates \(\AND_n\) within \(\frac{1}{20n}\) error requires degree \(\omega(\log(n))\).
\end{conjecture}

\noindent {\bf Organization of the Paper:}
We introduce some notation and preliminary statements required for our results in Section~\ref{sec:prelims}.
Section~\ref{sec:symlb} is dedicated to proving Theorem~\ref{thm:symcczlb}, and its extension to nested block symmetry.
In Section~\ref{sec:cdhapproach}, we describe our proposed approach for making progress on the constant degree hypothesis by improving lower bounds for torus polynomials approximating \(\AND\).
Finally, in Section~\ref{sec:periodub}, we prove Theorem~\ref{introthm:periodapprox} about symmetric torus polynomials approximating periodic functions.

\section{Preliminaries}
\label{sec:prelims}

We start by formally defining the notion of symmetry for circuits.

\begin{definition}[Symmetric Circuits]
For a subgroup \(\Gamma \leq S_n\), a circuit \(C\) over variables \((x_1, \ldots, x_n)\) is called to be \(\Gamma\)-symmetric if it satisfies the following property.
For each permutation \(\pi \in \Gamma\), there is a permutation \(\pi'\) over the gates in \(C\) that extends \(\pi\), that is \(\pi'(x_i) = \pi(x_i)\) for each \(i \in [n]\), such that the following holds:
A gate \(g\) is connected to another gate \(g'\) in \(C\) if and only if \(\pi'(g)\) is connected to \(\pi'(g')\) in \(C\).

If \(\Gamma = S_n\), the circuit is called symmetric.
\end{definition}

The circuits we consider are composed of \(\modm\) gates, defined as follows.
\[
\modm(x_1, \ldots, x_n) = 
\begin{cases}
1 & \sum_i {x_i} \mod m = 0 \\
0 & \sum_i {x_i} \mod m \neq 1
\end{cases}
\]

\noindent {\bf Torus Polynomials:}
The main tool for our results is approximation by torus polynomials, introduced in~\cite{bhrushundi2019torus}. We define it formally below.

\begin{definition}[Torus Polynomials (\cite{bhrushundi2019torus})]
Fix a natural number \(n\), and a Boolean function \(f : \{0, 1\}^n \to \{0, 1\}\).
Consider a polynomial \(P \in \mathbb{R}[x_1, \ldots, x_n]\).
We define \(P\) as a \emph{torus polynomial} approximating \(f\) within an error of \(\varepsilon\), for some \(0 \leq \varepsilon < \frac{1}{4}\), if the following holds for some integer function \(Z : \{0, 1\}^n \to \mathbb{Z}\):
For each \(a \in \{0, 1\}^n\), \(P(a)\) belongs to the interval \(P(a) \in [Z(a) + f(a)/2 - \varepsilon, Z(a) + f(a)/2 + \varepsilon]\).

In other words, the \emph{fractional part} of \(P(a)\) is at most \(\varepsilon\) away from \(\frac{f(a)}{2}\).
\end{definition}

The proof of our main result crucially relies on symmetric torus polynomials, in which the approximating polynomial is symmetric, as formally defined below.

\begin{definition}[Symmetric Polynomial]
A polynomial \(P \in \mathbb{R}[x_1, \ldots, x_n]\) is symmetric if for any permutation \(\pi \in S_n\), \(P(x) = P(\pi \circ x)\) syntactically.
In other words, monomials of the same degree share the same coefficient.

In general, \(P\) is \(\Gamma\)-symmetric, for some subgroup \(\Gamma \leqslant S_n\), if \(P(x) = P(\pi \circ x)\) for any \(\pi \in \Gamma\).
In other words, monomials within the same orbit under \(\Gamma\) share the same coefficient.
\end{definition}

In~\cite[Theorem 7]{krishan2026lower}, the authors proved a lower bound on the degree of symmetric torus polynomials approximating the \(\AND\) function.
We state this result, a key component in our proof, below.

\begin{theorem}[\cite{krishan2026lower}]
\label{thm:symtoruslb}
Any symmetric torus polynomial that approximates \(\AND_n\) within \(\frac{1}{20n}\) error must have degree \(\widetilde{\Omega}\rb*{\sqrt{n}}\).
\end{theorem}

\section{Size Lower Bounds for Symmetric $\ccz$ Circuits}
\label{sec:symlb}

The proof starts by converting a symmetric \(\ccz\) circuit to a symmetric torus polynomial that approximates it within \(\frac{1}{20n}\) error.
Then, we use Theorem~\ref{thm:symtoruslb}~\cite{krishan2026lower}, which allows us to prove the size lower bound.
Formally, the conversion result is as follows.

\begin{lemma}
\label{lem:symccztorus}
Consider any symmetric \(\ccz\) circuit of depth \(h\) and size \(s\).
Then, there exists a symmetric torus polynomial \(P\) that approximates it within \(\frac{1}{20n}\) error, such that \(\mathop{deg}(P) \leq \log^{O(h)}(s)\).
\end{lemma}

\begin{proof}[Proof Sketch]
We start with a \(\ccz\) circuit \(C\) as considered in the statement.
Our plan for finding a torus polynomial that approximates this circuit proceeds in three steps.

\begin{enumerate}
\item Convert the \(\ccz\) circuit to a layered circuit such that \(\AND\) gates appear only at the bottom layer, and each layer consists of \(\modq\) gates for a single prime power \(q\).
The conversion is well-known, see~\cite{beigel1994acc} for reference.
We argue that this conversion can be performed while preserving symmetry, as shown in Lemma~\ref{lem:symcczconvert}.

\item Next, we use the torus polynomial construction from~\cite{bhrushundi2019torus} to obtain a torus polynomial that approximates \(C\) within \(\frac{1}{20n}\) error.
As our argument depends crucially on this construction, we describe it in Lemma~\ref{lem:symtorus}.
This yields a torus polynomial of degree \(\log^{O(h)}(s)\).

\item Finally, we argue that the conversion produces a symmetric polynomial when starting with a symmetric circuit.
We prove this formally in Lemma~\ref{lem:symtorus}.
\end{enumerate}

At the end, we obtain a symmetric torus polynomial that approximates the circuit, with the claimed degree.
\end{proof}

Using Lemma~\ref{lem:symccztorus}, we can complete the proof of our main result as follows.

\begin{proof}[Proof of Theorem~\ref{thm:symcczlb}]
Consider a symmetric \(\ccz\) circuit \(C\), with depth \(h\), size \(s\) and \(\AND\) gates of fan-in at most \(d = O(1)\) at the bottom.
Then, there exists a symmetric torus polynomial \(P\), of degree at most \(\log^{O(h)}(s)\) that approximates \(C\) within \(\frac{1}{20n}\) error.
Now, assume that \(C\) computes \(\AND_n\).
Using~\cite[Theorem 7]{krishan2026lower}, we get that \(\log^{O(h)}(s) \geq \widetilde{\Omega}\rb*{\sqrt{n}}\).
Hence, \(s(n) = 2^{\widetilde{\Omega}(n^{1/O(h)}))}\).
\end{proof}

\subsection{Converting \(\ccz\) circuits to Layered Form}

To start the proof of Lemma~\ref{lem:symccztorus}, we describe the procedure to convert a symmetric \(\ccz\) circuit into a symmetric \emph{layered modular circuit}.
We define layered modular circuits below.

\begin{definition}[Layered Modular Circuit]
   A circuit is called a layered modular circuit if it satisfies the following conditions:
   \begin{itemize}
       \item Each gate in the circuit is a \(\modg_q\) gate for some prime power \(q\), except possibly the bottom layer which can consist of \(\AND\) gates of constant fan-in.
       \item Each layer of the circuit uses \(\modg_q\) gates for a single \(q\).
       \item Different \(\modg_q\) gates use prime powers \(q\) for distinct primes.
   \end{itemize}
\end{definition}

We use a well-known procedure for this conversion, see~\cite{beigel1994acc} for reference.
Our main contribution is to argue that the conversion preserves symmetry.

\begin{lemma}
\label{lem:symcczconvert}
Consider a \(\ccz\) circuit of size \(s\) and depth \(h\).
Then, it can be transformed to an equivalent layered modular circuit of size \(\poly(s)\) and depth \(O(h)\).
Crucially, if the original circuit is symmetric, the produced circuit is symmetric as well.
\end{lemma}

\begin{proof}
Broadly speaking, the conversion involves the following three steps (see~\cite{beigel1994acc}):

\begin{itemize}
\item Simulate a \(\modm\) gate as an \(\AND\) gate over \(\modq\) gates, where each \(\modq\) uses a prime-power factor \(q\) of \(m\), and no two \(\modq\) gates use a power of the same prime.
\item Convert this to a layered circuit by introducing dummy \(\modg\) gates.
\item Bring all \(\AND\) gates to the bottom by switching \(\AND\) and \(\modg\) gates at each layer.
\end{itemize}

Now, we argue that each of these steps can be performed while maintaining symmetry.
This suffices for the proof of the statement.

First, consider a \(\modm\) gate \(g\), where \(m = \prod_{i=1}^f q_i\) such that each \(q_i\) is a prime power for a distinct prime, and denote its inputs by \((g_1, \ldots, g_s)\).
The gate \(g\) is replace by an \(\AND\) gate \(g_a\) over \(\modg_{q_i}\) gates \(g_{q_i}\), each of which take \((g_1, \ldots, g_s)\) as their inputs.
To see how symmetry is preserved, choose any permutation \(\pi \in S_n\), which extends to a permutation \(\pi'\) over the gates in the original circuit.
Apply the transformation described above to \(\pi'(g)\), which is connected to \((\pi'(g_1), \ldots, \pi'(g_s))\).
As \(\pi'(g)\) is also a \(\modm\) gate, it is replaced by an \(\AND\) gate \(g'_a\) over \(\modg_{q_i}\) gates \(g'_{q_i}\), each of which take \((\pi'(g_1), \ldots, \pi'(g_s))\) as their input.
Hence, the permutation \(\pi''\) over the gates in the new circuit maps \(\pi''(g_a) = g'_a\), \(\pi''(g_{q_i}) = g'_{q_i}\) and \(\pi''(g_i) = \pi'(g_i)\).
It is easy to check that \(\pi''\) is as required.

Second, to ensure that each layer contains \(\modq\) gates for a single prime power \(q\), we modify the transformation from the first step itself.
Instead of replacing the \(\modm\) gate with an \(\AND\) over \(\modg_{q_i}\) gates, we replace it with a depth \(f+1\) gadget.
The top gate in this gadget is an \(\AND\) gate.
Now, we order \(q_i\)s in their increasing order, and form the next \(f\) layers using \(\modg_{q_i}\) gates in the \(i\)\textsuperscript{th} layer.
In the first layer, the leftmost \(\modg_{q_1}\) gate has fan-in \(s\), its \(j\)\textsuperscript{th} input being a path consisting of \(\modg_{q_2}, \ldots, \modg_{q_f}\) gates with \(g_j\) as the input to the bottom \(\modg_{q_f}\) gate.
The remaining \(f-1\) gates on the right have fan-in \(1\) each, taking a single gate from the second layer as input, as described inductively.

In the \(i\)\textsuperscript{th} layer, we describe the leftmost \(\modg_{q_i}\) gate after ignoring the gates already introduced by the layers above.
The leftmost \(\modg_{q_i}\) gate has fan-in \(s\), its \(j\)\textsuperscript{th} input being a path consisting of \(\modg_{q_{i+1}}, \ldots, \modg_{q_f}\) gates with \(g_j\) as the input to the bottom \(\modg_{q_f}\) gate.
All the other gates to its right have fan-in \(1\), each taking a single gate from the \(i+1\)\textsuperscript{th} layer as its input.
This completes the description of the gadget.
Clearly, the gadget is equivalent to \(\modm(g_1, \ldots, g_s)\), and size \(O(f^2s) = O(s)\).
The argument for why it preserves symmetry is similar to the argument for the first step.

Finally, for the third step, we need to bring \(\AND\) gates to the bottom.
This requires multiple applications of a well-known operation~\cite{beigel1994acc} which switches from \(\AND \circ \modq\) to \(\modq \circ \AND\).
Consider an \(\AND\) gate, over \(\modq\) gates \((g_1, \ldots, g_\ell)\), with \((h_{i, j_1}, \ldots, h_{i, j_{f'}})\) as the inputs for \(g_i\).
Here, we assume without loss of generality that each \(g_i\) has the same fan-in, by inserting ``dummy'' \(0\) inputs wherever necessary. 
Now, represent each \(\modq\) gate \(g_i\) as a polynomial \(P_i\) over \(\mathbb{Z}_{q}\) such that \(g_i(h_{i, j_i}, \ldots, h_{i, j_{f'}}) = P_i\rb*{\sum_{j = j_1}^{j_{f'}} h_{i, j}}\).
Then, write the expression for \(\AND(g_1, \ldots, g_\ell)\) as \(\prod_{i=1}^\ell P_i\).
Expand the product to get a single polynomial \(P\), such that each monomial of \(P\) becomes an \(\AND\) gate, feeding into a single \(\modq\) gate.
This completes the description of the operation to switch from \(\AND \circ \modq\) to \(\modq \circ \AND\).

We argue that each such operation preserves symmetry.
Consider an \(\AND\) gate \(g\), over \(\modq\) gates \((g_1, \ldots, g_\ell)\), with \((h_{i, j_1}, \ldots, h_{i, j_{f'}})\) as the inputs for \(g_i\).
As the original circuit is symmetric, for any permutation \(\pi \in S_n\), we have an extended permutation \(\pi'\) such that \(\pi'(g)\) takes \((\pi'(g_1), \ldots, \pi'(g_\ell))\) as its inputs, with \((\pi'(h_{i, j_1}), \ldots, \pi'(h_{i, j_{f'}}))\) as the inputs for \(\pi'(g_i)\).
Now, the polynomial \(P_i\) obtained for \(g_i\) is the same as \(P'_i\) obtained for \(\pi'(g_i)\).
Hence, for each monomial obtained in \(\prod_{i=1}^{\ell} P_i\) of the form \(\prod_{(i, j) \in S} h_{i, j}\) for some subset \(S \subseteq [\ell] \times [j_1, \ldots, j_{f'}]\), there is a corresponding monomial \(\prod_{(i, j) \in S} \pi'(h_{i, j})\).
Therefore, we can find a new permutation \(\pi''\), that maps the \(\AND\) gate corresponding to \(\prod_{(i, j) \in S} h_{i, j}\) with the \(\AND\) gate corresponding to \(\prod_{(i, j) \in S} \pi'(h_{i, j})\).
It is easy to verify that the new permutation \(\pi''\) suffices to prove symmetry for the new circuit.

Finally, one can easily verify that the size of the transformed circuit is at most \(O(\poly(s))\), and the depth is \(O(h)\).
This completes the proof.
\end{proof}

\subsection{Torus Polynomial Approximations for Layered $\ccz$ Circuits}
As the next step in the proof of Lemma~\ref{lem:symccztorus}, we describe the procedure to obtain a torus polynomial approximating a layered circuit.

\begin{lemma}
\label{lem:ccztorus}
Consider a layered modular circuit of depth \(h\) and size \(s\).
Then, there exists a torus polynomial \(P\) that approximates it within \(\frac{1}{20n}\) error, such that \(\mathop{deg}(P) \leq \log^{O(h)}(s)\).
\end{lemma}

The proof of this statement is implicit in~\cite{bhrushundi2019torus}; however, we reproduce the details of the proof as they are crucial for our argument.
As the first step, given a layered modular circuit \(C\), we construct an integer polynomial \(Q\), such that the binary expansion of \(Q(a)\) contains the value of \(C(a)\) at a fixed index.
Formally, we prove the following.

\begin{lemma}
\label{lem:cczintpoly}
Consider a layered circuit of depth \(h\) and size \(s\).
Then for any \(e \geq 1\), there exists an integer polynomial \(Q\) of degree \((e\log(s))^{O(h)}\) such that the following holds for some \(\ell \geq e\):
\[\forall a \in \{0, 1\}^n, Q(a) = C(a) 2^{\ell} + E(a) \mod 2^{\ell + e}\]
where \(E(a) \leq 2^{\ell - e}\).
\end{lemma}

For now, assume that the statement is true.
We can complete the proof of Lemma~\ref{lem:ccztorus} as follows.

\begin{proof}[Proof of Lemma~\ref{lem:ccztorus}]
    For a circuit \(C\) as per the statement, construct a polynomial \(Q\) for \(e = \ceil*{\log(10 n)}\) using Lemma~\ref{lem:cczintpoly}.
    As per the construction, for some \(\ell \geq e\), we have \(Q(a) = C(a) 2^{\ell} + E(a) \mod 2^{\ell + e}\) with \(E(a) \leq 2^{\ell - e}\).
    Define the torus polynomial \(P(x) = \frac{Q(x)}{2^{\ell+1}} \mod 1\).
    Then, we have
    \[P(a) = \frac{C(a)}{2} + \frac{E(a)}{2^{\ell + 1}} \mod 1\]
    The error for \(P\) is at most \(\frac{E(x)}{2^{\ell + 1}} \leq 2^{-e-1} \leq \frac{1}{20 n}\).
    Finally, \(P\) has the same degree as \(Q\), which is at most \((e \log(s))^{O(h)} = \log^{O(h)}(s)\).
    This completes the proof.
\end{proof}

Now, we prove Lemma~\ref{lem:cczintpoly}.
The proof idea is from~\cite{bhrushundi2019torus}, based on an induction on the circuit depth.
For a depth-\(1\) circuit, the statement is as follows.

\begin{lemma}[\cite{bhrushundi2019torus}]
\label{lem:modtorus}
For a \(\modq\) gate, where \(q\) is a prime power, there exists an integer polynomial \(Q\) of degree \(O(qe)\), and an integer \(\ell = O(qe\log(n))\), such that for each \(a \in \{0, 1\}^n\):
\[Q(a) = \modq(a) 2^\ell + E(a) \mod 2^{\ell+e}, 0 \leq E(a) \leq 2^{\ell-e}\]
\end{lemma}

We present the proof in the appendix for completeness, as it is not present in the published version of~\cite{bhrushundi2019torus}.
The inductive proof now proceeds as follows.

\begin{proof}[Proof of Lemma~\ref{lem:cczintpoly}]
The base case is a depth \(1\) circuit, which is covered in Lemma~\ref{lem:modtorus}.
Now, as the inductive hypothesis, assume that the statement is true for circuits of depth \(h\).
For the inductive step, consider a circuit of depth \(h+1\).

Suppose the root is a \(\modq\) gate, with \((y_1, \ldots, y_t)\) as its inputs, \(t \leq s\), where \(y_i = C_i(x)\) for the sub-circuit \(C_i\) of depth \(h\) over the original inputs \(x\).
Construct a polynomial \(Q\) of degree \(O(qe)\), such that \(Q(\cdot) = \modq(\cdot) 2^\ell + E(\cdot) \mod 2^{\ell+e}\) with \(\ell = O(qe \log(t))\) and \(E(\cdot) \leq 2^{\ell - e}\), using Lemma~\ref{lem:modtorus}.
Each monomial in this polynomial is a product of at most \(O(qe)\) many \(y_i\)'s.
The total number of such monomials is at most \(t' = t^{O(qe)}\).

Now, each monomial can be seen as an \(\AND\) over the corresponding sub-circuits \(C_i\), which can be propagated to the leaves.
Finally, we get a sub-circuit \(D_j(A_j)\) corresponding to each monomial of \(Q\), where \(A_j\) is a product of \(O(qe)\) many variables.
For each \(D_j\), use the inductive hypothesis to obtain a polynomial \(Q_j\), such that
\[Q_j(A_j(x)) = D_j(A_j(x)) 2^{\ell'} + E_j(A_j(x)) \mod 2^{\ell' + e'}\]
where \(E_j(A_j(x)) \leq 2^{\ell' - e'}\).
We choose the value for \(e'\) later.

Now, we have
\[
\sum_j D_j (A_j(x)) = C(x) 2^{\ell} + E(x) \mod 2^{\ell + e}
\]
Substituting the values of \(D_j\), we get
\begin{align*}
\sum_j Q_j(A_j(x)) &= \rb*{\sum_j D_j(A_j(x))} 2^{\ell'} + E'(x) \mod 2^{\ell' + e'} \\
 &= C(x) 2^{\ell + \ell'} + 2^{\ell'} E(x) + E'(x) \mod 2^{\ell + \ell' + e}
\end{align*}
To ensure the correctness of the latter equality, we will choose \(e'\) such that \(\ell' + e' \geq \ell + \ell' + e\).

Now, as per the claimed statement, we need to show that \(E'(x) + 2^{\ell'} E(x) \leq 2^{\ell + \ell' - e}\).
However, this does not follow directly, as it can be the case that \(E(a) = 2^{\ell - e}\) and \(E'(a) > 0\) for some \(a \in \{0, 1\}^n\).
Nonetheless, note that \(E'(x) \leq t' 2^{\ell' - e'}\).
Hence, if we choose \(e' \geq \ceil*{\log(t')}\), we get \(E'(x) + 2^{\ell'} E(x) \leq 2^{\ell' - O(e\log(s))} + 2^{\ell + \ell' - e} \leq 2^{\ell + \ell' - e + 1}\).
Therefore, if we start the proof by constructing \(Q\) with \(e+1\) when using Lemma~\ref{lem:modtorus}, we can ensure \(E'(x) + 2^{\ell'} E(x) \leq 2^{\ell + \ell' - e}\).

To finish the proof, first note that the choice of \(e' = \ceil*{\log(t')}\) gives us \(e' = O(qe \log(s))\).
This is sufficient to ensure \(\ell' + e' \geq \ell + \ell' + e\).
For the chosen \(e'\), the degree of each \(Q_j\) is \(O\rb*{(e\log(s))^{h-1}}\).

The final polynomial is \(\sum_j Q_j(A_j(x))\), with each \(A_j\) being a product of \(O(qe)\) many variables.
Hence, the degree bound we obtain is \(O\rb*{e\log(s)}^h\), as claimed in the statement.
This completes the proof.
\end{proof}

\subsection{Symmetric Circuits Lead to Symmetric Torus Polynomials}
The final step in the proof of Lemma~\ref{lem:symccztorus} is to prove that the above construction produces a symmetric torus polynomial for a symmetric \(\ccz\) circuit.
We actually prove a more general statement about \(\Gamma\)-symmetric circuits for an arbitrary subgroup \(\Gamma \leqslant S_n\).

\begin{lemma}
\label{lem:symtorus}
    For a \(\Gamma\)-symmetric circuit, the procedure in Lemma~\ref{lem:ccztorus} produces a \(\Gamma\)-symmetric torus polynomial approximation.
\end{lemma}

\begin{corollary}
\label{corr:symtorus}
For a symmetric layered modular circuit \(C\), the procedure in Lemma~\ref{lem:ccztorus} can be used to obtain a symmetric torus polynomial approximating the function computed by the circuit.
\end{corollary}

The proof of Lemma~\ref{lem:symtorus} follows from a slightly more general statement.
We describe how the construction behaves with respect to morphisms over circuits that preserve the edge connections.
Informally, we show that the construction produces syntactically similar torus polynomials for circuits that are syntactically similar.
The formal statement follows.

\begin{lemma}
\label{lem:permtorus}
    Consider two layered circuits \(C = (G, W)\) and \(C' = (G', W')\), with inputs \((x_1, \ldots, x_n)\) and \((y_1, \ldots, y_n)\) respectively, of the same depth.
    In the \(i\)\textsuperscript{th} layer of both circuits, they use \(\modg_{q_i}\) gates for the same \(q_i\).
    Then, for any map \(\rho : G \to G'\) with \(\rho(x_i) = y_i\), such that \((g, g') \in W\) if and only if \((\rho(g), \rho(g')) \in W'\), the following holds:
    If Lemma~\ref{lem:ccztorus} produces \(P(x)\) as a torus polynomial approximating \(C\), then it produces \(P(y)\) as the torus polynomial approximating \(C'\).
\end{lemma}

\begin{proof}
We argue this inductively based on the depths of the circuits \(C, C'\).
In the base case, with depth \(1\), we look at the proof of Lemma~\ref{lem:modtorus}.
The construction of the polynomial \(Q\) depends only on the fan-in of the \(\modq\) gate and the error parameter \(e\).
Hence, the base case follows.

For the inductive case, we look at the proof of Lemma~\ref{lem:ccztorus}.
The polynomials \(Q, Q'\) we construct for the top \(\modq\) gate, in \(C, C'\) respectively, only depend on the fan-in.
Hence, for each monomial in \(Q\), we can find a corresponding monomial in \(Q'\) using the map \(\rho\).

Now, consider a monomial \(M_j\) of \(Q\) and the corresponding monomial \(M'_j\) from \(Q'\) obtained using \(\rho\).
We consider the monomials as \(\AND\) over their inputs, then push the \(\AND\) gate to the bottom, to obtain circuits of form \(D_j(A_j(x))\) and \(D'_j(A'_j(y))\).
Using induction hypothesis, the torus polynomials we obtain for \(D_j(A_j(x))\) and \(D'_j(A_j(y))\), denoted by \(Q_j(A_j(x))\) and \(Q'_j(A'_j(y))\) respectively, are equal.
Therefore, we have \(\sum_j Q_j(A_j(x)) = \sum_j Q'_j(A'_j(y))\), where the LHS and RHS are the torus polynomials produced for \(C\) and \(C'\) respectively.
This completes the proof of the statement.
\end{proof}

Now, we can finish the proof of Lemma~\ref{lem:symtorus} as follows.

\begin{proof}[Proof of Lemma~\ref{lem:symtorus}]
    Consider a \(\Gamma\)-symmetric layered circuit \(C = (G, W)\) as per the statement.
    Any permutation \(\pi \in \Gamma\) extends to a permutation \(\pi' \in \aut(G)\), such that \((g, g') \in W\) if and only if \((\pi'(g), \pi'(g')) \in W\).
    Now, construct a torus polynomial \(P(x)\) that approximates \(C\) using Lemma~\ref{lem:ccztorus}.
    Then, Lemma~\ref{lem:permtorus} implies that \(P(\pi(x))\) is produced as the torus polynomial approximating \(C_{\pi'} = (\pi'(G), W)\).
    
    However, \(C_{\pi'} = C\), hence, Lemma~\ref{lem:ccztorus} would produce \(P(x)\) as the torus polynomial approximating \(C\).
    Therefore, \(P(\pi(x)) = P(x)\), completing the proof that \(P\) is symmetric.
\end{proof}

\subsection{Nested Block Symmetric Groups}

Now, we consider a weaker notion of symmetry, studied by Pago~\cite{pago2026optimal}, which they called nested block symmetry.
Nested block symmetric groups are automorphism groups of rooted trees, defined formally as follows.

\begin{definition}[Nested Block Symmetric Group]
    Consider a tree of depth \(h\), such that each node at distance \(i\) from the root has exactly \(k_i\) many children, and exactly \(n\) leaves labeled with \(x_1, \ldots, x_n\).
    Note that this implies \(\prod_{i=1}^h k_i = n\).
    Denote this tree by \(\mathcal{T}^{\mathbf{k}}_n\), where \(\mathbf{k}\) denotes the tuple \((k_1, \ldots, k_h)\).
    The nested block symmetric group \(\Gamma^{\mathbf{k}}_n\) is defined as the group of permutations \(\pi \in S_n\) over the leaves such that some \(\pi' \in \aut(\mathcal{T}^{\mathbf{k}}_n)\) extends \(\pi\).
\end{definition}

We prove a lower bound for \(\Gamma^{\mathbf{k}}_n\)-symmetric \(\ccz\) circuits computing \(\AND\).

\begin{theorem}
    \label{thm:nestedsymlb}
    Any \(\Gamma^{\mathbf{k}}_n\)-symmetric depth-\(h\) \(\ccz\) circuit, with an underlying tree \(\mathcal{T}^{\mathbf{k}}_n\), requires \(2^{\widetilde{\Omega}(k^{1/O(h)})}\) size to compute \(\AND\), where \(k = \max_{k' \in \mathbf{k}} k'\).
\end{theorem}

\begin{proof}
    To start the proof, we use Lemma~\ref{lem:symtorus} to construct a \(\Gamma^{\mathbf{k}}_n\)-symmetric torus polynomial \(P\) of degree \(\log(s)^{O(h)}(n)\) approximating the circuit within an error of \(\varepsilon=\frac{1}{20n}\).
    Denote the depth of \(\mathcal{T}^{\mathbf{k}}_n\) as \(h'\).
    Now, consider the case when \(k = k_{h'}\), i.e. the maximum fan-in in the underlying tree \(\mathcal{T}^{\mathbf{k}}_n\) appears at the layer just above the leaves.
    Then, we choose an arbitrary node above the leaves and set all variables that occur outside this node to \(1\).
    After applying this partial restriction, we notice that the group \(\Gamma^{\mathbf{k}}_n\) collapses to the symmetric group \(S_{k}\).
    Hence, we get a symmetric torus polynomial \(P'\) approximating \(\AND\) over \(k\) variables.
    Therefore, using Theorem~\ref{thm:symtoruslb} for \(\AND\) over \(k\) variables, we get the size lower bound.

    Otherwise, if \(k \neq k_{h'}\), we reduce the depth of \(\mathcal{T}^{\mathbf{k}}_n\) by considering each node just above the leaves, and setting all variables appearing within such a node as equal to each other.
    After applying this partial restriction, we effectively reduce the depth of \(\mathcal{T}^{\mathbf{k}}_n\) by one.
    Note that this transformation does not increase the degree of \(P'\), and the function remains \(\AND\) over the remaining variables.
    Now, we can proceed inductively till we reach the level where \(k\) appears, and use our argument above to finish the proof.
\end{proof}

\section{Towards Constant Degree Hypothesis for Semiprime Moduli}
\label{sec:cdhapproach}

In this section, we propose an approach to prove size lower bounds for \(\modp \circ \modm \circ \AND_{O(1)}\) circuits computing \(\AND_n\), where \(m\) has two distinct prime divisors with one of them being \(p\).
For the sake of simplicity, we state our results for \(m = pq\) being a product of two primes.
The approach is again based on torus polynomials: we construct a low-degree torus polynomial approximation for such circuits, albeit in a slightly modified sense.
We consider a generalized version of torus polynomial approximation, also studied in~\cite{bhrushundi2019torus}, wherein the polynomial is close to \(\alpha f\) for some \(\alpha \in (0, 1)\).
For clarity, we define it formally below.

\begin{definition}[\(\alpha\)-Torus Polynomial Approximation]
For some \(\alpha \in (0, 1)\), and a Boolean function \(f : \{0, 1\}^n \to \{0, 1\}\), we define \(P\) as an \(\alpha\)-torus polynomial approximating \(f\) within an error of \(\varepsilon\), if the \emph{fractional part} of \(P(a)\) is at most \(\varepsilon\) away from \(\alpha f(a)\) for each \(a \in \{0, 1\}^n\).
\end{definition}

For our main result, we need the following statement, which follows from a simple modification of the proof for Lemma~\ref{lem:modtorus}.

\begin{lemma}
\label{lem:alphatorus}
For a \(\modq\) gate, where \(q\) is a prime power, and any natural number \(p \geq 2\), there exists an integer polynomial \(Q\) of degree \(O(qe)\), and an integer \(\ell = O(qe\log(n))\), such that for each \(a \in \{0, 1\}^n\):
\[Q(a) = \modq(a) p^\ell + E(a) \mod p^{\ell+e}, 0 \leq E(a) \leq p^{\ell-e}\]
\end{lemma}

Now, we proceed with the proof of our main result.

\begin{theorem}
\label{thm:cdhapprox}
For any \(\modp \circ \modg_{pq} \circ \AND_{O(1)}\) circuit of size \(s = \poly(n)\), there exists a \(\frac{1}{p}\)-torus polynomial of degree \(O(\log(n))\) that approximates the circuit within \(\frac{1}{20n}\) error.
\end{theorem}

\begin{proof}
To start the proof, we rewrite the \(\modg_{pq}\) gate as an \(\AND\) over \(\modp\) and \(\modq\) gates.
Then, we transform the circuit such that the first two layers are \(\modp\) gates, then a layer of \(\modq\) gates, and finally \(\AND\) gates of \(O(1)\) fan-in at the bottom.
Now, consider the top two layers of \(\modp\) gates as a circuit \(C\) with inputs \((y_1, \ldots, y_m)\), where each \(y_i\) is computed by a \(\modq \circ \AND_{O(1)}\) circuit.
Construct a polynomial \(P\), using Fermat-Euler theorem, such that \(P(\cdot) = C(\cdot) \mod p\).
The polynomial \(P\) has \(O(1)\) degree, hence, it contains at most \(t' = \poly(n)\) many monomials.
Each of these monomials \(M_j\) can be thought of as an \(\AND_{O(1)} \circ \modq \circ \AND_{O(1)}\) circuit, which we convert to a \(\modq \circ \AND_{O(1)}\) circuit.

Now, for each \(\modq\) gate \(D_j\) with its inputs being \(A_j(x)\), we use Lemma~\ref{lem:alphatorus} to obtain a polynomial \(P_j\) such that \(P_j(A_j(x)) = D_j(A_j(x)) p^\ell + E_j(A_j(x)) \mod p^{\ell + 1}\), where \(E_j(A_j(x)) \leq p^{\ell - e}\).
We will choose the value for \(e\) later.
Then, we get \(\sum_{j=1}^{t'} P_j(A_j(x)) = \rb*{\sum_{j=1}^{t'} D_j(A_j(x))} p^\ell + t' E(x) \mod p^{\ell + 1} = C(x) p^\ell + t' E(x) \mod p^{\ell + 1}\).
Finally, we need to ensure that \(t' E(x) \leq \frac{p^{\ell}}{20n}\), for which it suffices to choose a large enough \(e = O(\log(n))\).
The final polynomial as required by the statement is \(\frac{\sum_{j=1}^{t'} P_j(A_j(x))}{p^{\ell+1}}\), which has degree \(O(\log(n))\).
\end{proof}

This result suggests a potential approach for proving that \(\modp \circ \modg_{pq} \circ \AND_{O(1)}\) circuits require superpolynomial size to compute \(\AND_n\).
In~\cite[Theorem 6]{krishan2026lower}, the authors proved that any torus polynomial that approximates \(\AND_n\) within an error smaller than \(\varepsilon\) requires degree more than \(\log\rb*{\frac{1}{\varepsilon}} - 2\).
For \(\varepsilon = \frac{1}{20n}\), the degree required is \(\Omega(\log(n))\).
Their proof can be easily extended to prove that any \(\alpha\)-torus polynomial that approximates \(\AND_n\) within \(\frac{1}{20n}\) error requires degree \(\Omega(\log(n))\).
The authors note a gap between this lower bound and the known upper bound, and suggest bridging this gap in~\cite[Open Problem 4]{krishan2026lower}, hinting at a possibility of improving the lower bound.
Using Theorem~\ref{thm:cdhapprox}, we note that improving the lower bound to \(\omega(\log(n))\) proves the superpolynomial size lower bound we are aiming for.
Hence, we conjecture that the lower bound is indeed stronger.

\begin{conjecture}
\label{conj:genandlb}
    For any \(\alpha \in (0, 1)\), any \(\alpha\)-torus polynomial that approximates \(\AND_n\) within \(\frac{1}{20n}\) error requires degree \(\omega(\log(n))\).
\end{conjecture}

\section{Degree Upper Bounds for Periodic Functions}
\label{sec:periodub}

In this section, we establish an upper bound on the degree of symmetric torus polynomials approximating periodic functions.
We prove Theorem~\ref{introthm:periodapprox} restated below.
\thmperiodapprox*

As the first step in the proof, we construct a symmetric layered modular circuit that computes \(\modm\).
The statement is a simple case of Lemma~\ref{lem:symcczconvert}, for which we omit the proof.

\begin{lemma}
    \label{lem:modmconvert}
    There is a symmetric layered modular circuit that computes \(\modm\).
    If \(m\) has \(O(1)\) distinct prime divisors, then the circuit has size \(O(n)\) and depth \(O(1)\).
\end{lemma}

Next, we describe a construction for torus polynomials approximating layered modular circuits.
The proof is a simple modification of the proof for Lemma~\ref{lem:ccztorus}, which we omit.

\begin{lemma}
    \label{lem:modmtorus}
    Consider a layered modular circuit of depth \(O(1)\) and size \(O(n)\) that uses \(\modg_{q_i}\) gates for \(i \in [h]\).
    Then, there exists a torus polynomial of degree \(\prod_i q_i \cdot (e\log(n))^{O(1)}\) approximating the circuit within \(\varepsilon = 2^{-e}\) error.
\end{lemma}

Using this result, we can prove Theorem~\ref{introthm:periodapprox} as follows.

\begin{proof}[Proof of Theorem~\ref{introthm:periodapprox}]
    First, we prove the statement for \(f=\modm\) function, where \(m\) has \(r=O(1)\) distinct prime divisors.
    We use Lemma~\ref{lem:modmconvert} to obtain a symmetric layered modular circuit computing \(\modm\).
    Then, we use Lemma~\ref{lem:modmtorus} to obtain a torus polynomial of degree \(m \cdot (\log(n))^{O(1)}\) approximating \(\modm\) within \(\frac{1}{20n}\) error.
    Using Lemma~\ref{lem:symtorus}, we get that the polynomial is symmetric.
    This completes the \(\modm\) case.

    Now, consider any symmetric function \(f : \{0, 1\}^n \to \{0, 1\}\) of period \(m\).
    Any such function can be represented as \(\modm^A\) for some \(A \subseteq [m]\), defined as follows.
    \[
    \modm^A(x_1, \ldots, x_n) =
    \begin{cases}
        1 & \sum_i x_i \in A \\
        0 & \sum_i x_i \notin A
    \end{cases}
    \]
    We write \(\modm^A\) as a disjoint \(\OR\) of \(\modm^{\{a\}}\) for each \(a \in A\).
    Each \(\modm^{\{a\}}\) over \(n\) inputs can be rewritten as \(\modm\) over at most \(2n\) inputs with some dummy \(1\) inputs.
    Then, for each \(\modm^{\{a\}}\) gate, we obtain a symmetric torus polynomial \(P_{\{a\}}\) approximating it within \(\frac{1}{20n^2}\) error using Lemma~\ref{lem:modmtorus}.

    Finally, note that \(\modm^A\) is a disjoint \(\OR\) over \(\modm^{\{a\}}\).
    Hence, in the sum \(\sum_{a \in A} P_{\{a\}}\), at most one polynomial contributes a fractional part of \(\frac{1}{2}\).
    Moreover, each \(P_{\{a\}}\) contributes an error within \(\frac{1}{20n^2}\), yielding a total error of \(\frac{\abs{A}}{20n^2} \leq \frac{1}{20n}\).
    Therefore, \(\sum_{a \in A} P_{\{a\}}\) is a symmetric torus polynomial approximating \(\modm^A\) within an error of \(\frac{1}{20n}\).
    
    The degree of this polynomial is as claimed in the statement.
    This completes the proof.
\end{proof}

\section*{Discussion and Open Problems}
One immediate direction is to make progress on Conjecture~\ref{conj:genandlb}.
In this direction, we present evidence that the proof of~\cite[Theorem 6]{krishan2026lower} presented by the authors has room for improvement.
Formally, the statement we prove improves the degree lower bound by an additive factor of \(1\).
The improvement is more conceptual, even though the improvement on the degree bound is minor.
It highlights that there is room for improvement in their proof, and further analysis may lead to a resolution of Conjecture~\ref{conj:genandlb}.
Note that we borrow heavily from the notation setup in~\cite{krishan2026lower} for brevity.

\begin{theorem}
    Any torus polynomial approximating \(\AND\) within an error smaller than \(\varepsilon\) must have degree more than \(\log\rb*{\frac{1}{\varepsilon}} - 1\).
\end{theorem}

\begin{proof}
Fix some \(n, d\), denote \(f = \AND_n\), and choose any \(\varepsilon < \frac{1}{2^{d+1}}\).
We use the method described in~\cite[Theorem 8]{krishan2026lower} for the proof.
To prove the lower bound, we need to find a witness \(\gamma \in \nulls(M(n, d))\) for each \(Z : \{0, 1\}^n \to \mathbb{Z}\) such that:
\begin{equation}
\label{eqn:nulls}
\abs*{\ip*{Z + \frac{f}{2}, \gamma}} > \varepsilon \norm{\gamma}_1
\end{equation}

First, we use~\cite[Lemma 9]{krishan2026lower}, which implies that \(Z(a) = 0\) for any \(a \in \{0, 1\}^n\) with \(\abs{a} \leq d\).
Now, choose any point with \(\abs{a} = d+1\).
Construct a vector \(\gamma \in \nulls(M(n, d))\) using~\cite[Construction 2]{krishan2026lower} with \(S_1 = \emptyset, S_2 = a\) and \(I = [d+1]*\) as the input to the construction.
For this vector, we have \(\norm{\gamma}_1 = 2^{d+1}\), \(\gamma_a = 1\), and \(\gamma_{a'} = 0\) for any point \(\abs{a'} \geq d+1\).
Hence, \(\ip*{Z + \frac{f}{2}, \gamma} = Z(a) \gamma_a = Z_a\), whereas \(\varepsilon \norm{\gamma}_1 < 1\).
Therefore, if \(\abs{Z(a)} \geq 1\), we have that inequality~\ref{eqn:nulls} is satisfied.
This leaves us to consider \(Z(a) = 0\) as the only remaining possibility.
Note that this argument applies to each point \(a\) with \(\abs{a} = d+1\) independently to show that \(Z(a) = 0\) for that point.

We continue with the argument inductively, stopping before we reach Hamming weight \(n\), assuming that \(Z(a') = 0\) for all points with \(\abs{a'} = i\) for some \(i \geq d+1\).
Then, consider a point \(a\) with \(\abs{a} = i+1\), and invoke~\cite[Construction 2]{krishan2026lower} with \(S_1 = a'', S_2 = a\) and \(I = [d+1]^*\) for some \(a'' \subseteq a\) of size \(\abs{a''} = \abs{a} - (d+1)\).
Again, as in the base case, if \(\abs{Z(a)} = 1\), then inequality~\ref{eqn:nulls} is satisfied.
Hence, \(Z(a) = 0\) is the only possibility that remains.

Finally, once we reach Hamming weight \(n\) with \(a = 1^n\), we have \(Z(a') = 0\) for any point with \(\abs{a'} < n\).
Now, as the integer parts up to Hamming weight \(n-1\) are \(0\), we can think of the torus polynomial as a real polynomial approximating the following function \(g\):
\[
g(a) = \begin{cases}
0 & a \neq 1^n \\
Z(a) + \frac{1}{2} & a = 1^n
\end{cases}
\]
However, any real polynomial that approximates \(g\) within \(\varepsilon\) error also approximates \(\AND\) within at most \(\varepsilon\) error.
As the real approximation degree of \(\AND\) for \(\varepsilon\) error is \(\Omega\rb*{\sqrt{n \log(\frac{1}{\varepsilon})}}\)~\cite{BT20}, such a polynomial cannot exist.
This completes the proof.
\end{proof}

The conceptual contribution of the proof above is to use multiple vectors to witness the lower bound.
In~\cite{krishan2026lower}, the authors had used a single witness for each \(Z\) to prove their result.
Hence, this can be thought of as a step towards further improvements.

In another direction, the proof technique of Theorem~\ref{thm:nestedsymlb} works for any subgroup \(\Gamma \leq S_n\) which simplifies to a large enough copy of the symmetric group under two restrictions: setting variables equal to each other, and setting variables as \(1\).
In particular, consider \(\Gamma\) that collapses to \(S_{n'}\) for some \(n'\) that is superpolynomial in \(\log(n)\).
Then, using our arguments, one obtains a superpolynomial lower bound on the size of \(\Gamma\)-symmetric \(\ccz\) circuits computing \(\AND\).
We leave it open to describe such subgroups \(\Gamma \leq S_n\).

\begin{problem}
    Describe subgroups \(\Gamma \leq S_n\) such that under the following restrictions:
    \begin{itemize}
        \item setting variables equal to each other,
        \item setting variables as \(1\),
    \end{itemize}
    the group collapses to \(S_{n'}\) for some \(n' = \log^{\omega(1)}(n)\).
\end{problem}

\noindent
\textbf{Acknowledgments:} We would like to thank anonymous MFCS 2026 reviewers for their helpful and detailed comments.
\appendix

\addtocontents{toc}{\protect\setcounter{tocdepth}{-1}}

\section{Appendix: Proof of Lemma~\ref{lem:modtorus}}

\begin{proof}
Let \(T\) be a polynomial such that \(T(x) = \modq(x) \mod q^{\ell'}\) for some \(\ell'\) to be chosen later.
Here, \(T\) is constructed by combining Fermat-Euler theorem and modulus amplifying polynomials~\cite{beigel1994acc}, which has degree \(O(\ell')\).
Define \(Q = \rb*{w \ceil*{\frac{2^\ell}{q^{\ell'}}} + 2^\ell} T\) with \(w = -q^{\ell'} \mod 2^e\).

The calculations are the same as those communicated to us by the authors in~\cite{bhrushundi2019torus}.
We reproduce them for the sake of completeness.
Fix some \(x\), and let \(T(x) = aq^{\ell'} + \modq(x)\).
Now, we continue the calculation as follows:
\begin{align*}
    Q(x) &= \rb*{w \ceil*{\frac{2^\ell}{q^{\ell'}}} + 2^{\ell}} T(x) \\
    &= \rb*{w \rb*{\frac{2^\ell}{q^{\ell'}} + 1 - \left\{\frac{2^\ell}{q^{\ell'}}\right\}}} \rb*{aq^{\ell'} + \modq(x)} \\
    &= \rb*{wa2^{\ell} + q^{\ell'} a2^\ell} + 2^\ell \modq(x) + waq^{\ell'} - waq^{\ell'} \left\{\frac{2^{\ell}}{q^{\ell'}}\right\} + w \ceil*{\frac{2^{\ell}}{q^{\ell'}}} \modq(x)
\end{align*}

Now, note that \(\rb*{wa2^{\ell} + q^{\ell'}a 2^{\ell}}  \mod 2^{\ell + e} = 0\).
Moreover, \(\abs{aq^{\ell'}} \leq 2^{O(q\ell' \log(n))}\) and \(w \leq 2^e\).
Hence, if we ensure the following:
\[
w \rb*{aq^{\ell'} \rb*{1 - \left\{\frac{2^\ell}{q^{\ell'}}\right\}} + \ceil*{\frac{2^{\ell}}{q^{\ell'}}}} \leq 2^{\ell - e}
\]
Then, \(E(x) \leq 2^{\ell - e}\), as required by the statement.
Choosing a large enough \(\ell' = O(e)\) and \(\ell = O(qe \log(n))\) suffices to ensure that the calculation goes through.
\end{proof}

\bibliographystyle{alpha}
\bibliography{references}

\end{document}